\documentclass[manyauthors]{fundam}
\usepackage{amsmath,amstext,amssymb}

\usepackage[english]{babel}
\usepackage[T1]{fontenc}
\usepackage[utf8]{inputenc}
\usepackage{tikz}
\usepackage{url}
\usepackage{times}

\usepackage{hyperref}

\newtheorem{problem}{Problem}

\newcommand{\EXTRA}[1]{}
\usepackage{shuffle}
\newcommand{\ov}{\overline}
\newcommand{\Al}{\Sigma}
\newcommand{\eps}{\varepsilon}
\newcommand{\N}{\mathbb{N}}

\begin{document}


\setcounter{page}{269}
\publyear{24}
\papernumber{2182}
\volume{191}
\issue{3-4}

\finalVersionForARXIV


\title{Decision Problems on Copying and Shuffling}

\author{Vesa Halava\thanks{Supported by emmy.network foundation under the aegis of the Fondation de Luxembourg.}\thanks{Address for correspondence: Department of Mathematics and Statistics, University of Turku, Finland.}
\\
Department of Mathematics and Statistics\\
University of Turku, Finland\\
vesa.halava@utu.fi
\and Tero Harju \\
Department of Mathematics and Statistics\\
University of Turku, Finland\\
harju@utu.fi
\and Dirk Nowotka \\
Department of Computer Science\\
Kiel University, Germany \\
dn@zs.uni-kiel.de
\and Esa Sahla \\
Department of Mathematics and Statistics\\ University of Turku, Finland\\
etsahla@gmail.com   }

\runninghead{V. Halava et al.}{Decision Problems on Copying and Shuffling}

\maketitle

\begin{abstract}
We study decision problems of the form: given a regular or linear context-free language $L$,
is there a word of a  given fixed form in $L$, where
given fixed forms are based on word operations copy, marked copy,
shuffle and their combinations.
\end{abstract}


\section{Introduction}

We consider classic problems on decidability issues of formal languages.
We shall fill in gaps that have remained for elementary operations copy and shuffle and their variants on words and languages. The presented results, as well as the known results on the topic, are presented in the table affixed in the second page leaving two open cases for further studies.

\medskip
We investigate the decidability status of several special membership problems
for regular and linear context-free (linear CF) languages where it is asked whether or not the language contains a word of a certain form. Let $L$ be a given language.
The operations and the question are the following:
\begin{enumerate}
\itemsep=0.9pt
\item \emph{copy}, i.e., does there exist a square $ww\in L$ for some word $w$?

\item \emph{reversed copy}, i.e.,  does there exist a word $ww^r\in L$ for some word $w$, where
 $w^r$ denotes the reversal of the word $w$?

\item \emph{marked copy}, i.e.,  does there exist a word $w\ov{w}\in L$ for some word $w$,
where $\ov{w}$ denotes a marked copy of the word $w$? (For a definition of a marked copy, see page~\pageref{def:mc})

\item \emph{self-shuffle}, i.e.,  does there exist a word $u\in w\shuffle w$ with $u\in L$ for some word $w$, where $\shuffle$ denotes the shuffle operations of two words?

\item \emph{shuffle with reverse}, i.e.,  does there exist a word $u\in w\shuffle w^r$ with $u\in L$ for some word $w$?

\item \emph{marked shuffle}, i.e., does there exist
a word $u\in w\shuffle \bar{w}$ with $u\in L$ for some word $w$?
\end{enumerate}

The decidability statuses of these questions are listed in the following table, where
D (resp. U) means that the problem is \emph{decidable} (resp. \emph{undecidable}) and the question mark denotes problems that remain unsettled. After the symbols  D and U
we give a reference for the proof in the text.
Here Reg stands for the regular languages and Lin for the linear context-free languages.

\bigskip

\begin{minipage}{11pc}
\begin{tabular}{l|lcl}
&Reg&&Lin\\ \hline
$ww \in L$&D (Cor.~\ref{Reg:copyspes}) &\phantom{space}& U  (Thm.~\ref{Lin:copy})\\
$ww^r \in L$&D (Cor.~\ref{Reg:mirror}) && U (Thm.~\ref{Lin:revcopy})\\
$w\ov{w} \in L$&D (Thm.~\ref{Reg:mc})&&U (Thm.~\ref{Lin:marked})\\
$w \shuffle w \cap L \ne \emptyset$ &?&&U (Thm.~\ref{Lin:shuffle})\\
$w \shuffle w^r \cap L \ne \emptyset$&?&&U (Thm.~\ref{Lin:revshuf})\\
$w \shuffle \ov{w} \cap L \ne \emptyset$ &U (Thm.~\ref{Reg:marked}) &&U (Thm.~\ref{Lin:marshuff}) \\
\end{tabular}
\end{minipage}

\bigskip

We also study decidability of some special inclusion problems related to the above problems.
For example, we investigate the problem of whether a given regular, linear context-free or context-free language is closed under taking squares, and also, the problem whether the set all squares generated by another given language is a subset of a given language.

\EXTRA{
\begin{problem}
Prove that RE languages can be generated by 2-way CF grammars. These have symmetric productions $A \leftrightarrow \alpha$.
\end{problem}

E.~Friedman considered languages accepted by \emph{simple machines}. These are stateless (i.e., they have one state)
deterministic pushdown automata without $\varepsilon$-moves that accepted on empty stack.
She showed, TCS 1976,  that the inclusion problem is undecidable for
the languages of these automata.
}

There are naturally many related language operations to be investigated;
see Rampersad and Shallit \cite{Rampersad2010}, where
among other results it was shown that it is undecidable
whether a context-free grammar generates a square.
We deal with this problem in Theorem~\ref{Lin:copy} for linear CF-languages.

\section{Preliminaries}

Let $\Sigma$ be a finite alphabet. A \emph{word} over $\Sigma$ is a finite sequence of symbols of $\Sigma$. The \emph{empty word} is denoted by $\eps$. The \emph{length} of a word $w=a_1\cdots a_k$, where $a_i\in\Al$ for all $i=1,\dots, k$, is $k$
and it is denoted by $|w|$.
The set of all words over $\Sigma$ is denoted by $\Sigma^*$ and the set of all non-empty words by~$\Sigma^+$.

For two words $u, v\in \Sigma^*$,  their \emph{concatenation} is
$u\cdot v=uv$. A \emph{factorization} of a word $w\in \Sigma^*$ is a finite sequence $u_1,\dots, u_k$, where $u_i\in \Sigma^*$ for all $i$, such that $w=u_1\cdots u_k$. A word $u$ is a \emph{prefix} of a word $w$ if $w=uv$ for some word $v$.

The powers of words are defined inductively: $w^0=\eps$ and
for all $n\in \N$, $w^{n+1}=w^n\cdot w$.
We say that a word $w$ is \emph{primitive}, if for all $u\in \Al^*$,
$w=u^n$ implies that $n=1$.

Let $w^r$ denote the \emph{reversal} (or the \emph{mirror image}) of $w$, that is,
$w^r= a_n\cdots a_1$ for $w=a_1\cdots a_n$, where $a_i\in\Al$ for all $i$.

With an alphabet $\Al$ we accompany a \emph{marked copy alphabet} \label{def:mc}
$\bar{\Al} = \{\bar{a} \mid a \in \Al \}$ , where $\Al \cap \bar{\Al} = \emptyset$.
For a word $w=a_1a_2 \cdots a_n$, let
$\ov{w} = \bar{a}_1\bar{a}_2 \cdots \bar{a}_n$ be the \emph{marked copy} of $w$.

A subset of $\Al^*$ is called \emph{language}.
Denote by $L^c$ the complement of  $L$, that is, $L^c=\Al^* \setminus L$.

We assume that the reader is familiar with the basic notions of language theory; see e.g., Salomaa~\cite{Sa-73} for definitions of \emph{regular languages}, \emph{finite automata}, \emph{context-free (CF) languages}, \emph{pushdown automata},   \emph{context-sensitive languages} and \emph{pumping lemmas} for regular and context-free languages.

We briefly recall a few basic facts. First of all, recall that a language $L$ is a \emph{linear CF-language} if it is accepted by a pushdown automaton (PDA) that makes at most one reversal (from increasing to decreasing mode) on its stack. Equivalently, each linear CF-language is generated by a linear context-free grammar, where the productions have at most one non-terminal on the right hand side.

\begin{example}\label{LanE}
The language of all palindromes of even length,
\begin{equation*}
E = \{ww^r\mid w\in \Al^*\},
\end{equation*}
is a linear CF-language. Indeed,
$E$ can be accepted by a non-deterministic linear PDA which first reads symbols onto the stack
until it (non-deterministically) decides to check by popping symbols whether the rest of the input word agrees with the stack content.
\end{example}

We use the Pumping lemma for regular languages to show that certain languages are not regular.

\begin{lemma}\label{PLreg}
For a regular language $L$,  there exists a natural number $p\ge 1$ such that,
if $w \in L$ is of length $|w| \ge p$, then it has a factorization $w=xyz$
with $|y|\ge 1$ and $|xy|\le p$,  such that $xy^nz \in L$ for all $n\in\N$.
\end{lemma}

The Pumping lemma for CF-languages has two simulaneous pumps.

\begin{lemma}\label{PLCF}
For a CF-language $L$, there exists a natural number $p\ge 1$ such that,
if $w\in L$ has length $|w| \ge p$, then it has a factorization
$w=uvwxy$ with $|vx|\ge 1$ and $|vwx|\le p$,  such that
$uv^nwx^n y \in L$ for all $n\in\N$.
\end{lemma}

Next we define two special languages.
Firstly, let $P \subseteq \Al^*$ be a language. The \emph{copy language} of $P$
is defined as the set of all second powers of words of $P$:
\[
C_P= \{ww \mid w \in P\}.
\]
It is well-known that for a regular language $P$,
the copy language $C_P$ is context-sensitive, but not context-free (see, for example, \cite{Kutrib2018}). We state the following open problem concerning the copy languages.
A \emph{one-counter automaton} is a pushdown automaton, with a single stack letter, which is able to check the emptiness of the stack.

\begin{problem}
Is the complement of $C_{\Al^*}$ a one-counter language?
\end{problem}

Considering the marked copy, the problem becomes easier.

\begin{lemma}\label{MC:onec}
For a regular language $P$, the complement of the
marked copy language $\{w\ov{w} \mid w \in P\}$
is a one-counter language.
\end{lemma}

\begin{proof}
Sketch of the proof: Assume $P$ is accepted by a
finite automaton $A$, that is, $P=L(A)$.
Let $\bar{A}$ be the \emph{copied version} of $A$, i.e.,
where the letters in the transitions are changed to marked letters.

\medskip
We construct a one reversal nondeterministic PDA $B$ for $A$ and its marked copy automaton $\bar{A}$ as follows:
\begin{enumerate}
\item $B$ simulates $A$ and reads symbols from $\Al$ and adds one to the counter
to count the length of the prefix read so far.

\item At one point $B$ remembers the symbol $a\in \Al$ under its reading
head and stops writing to the stack.
$B$ continues by reading the rest of the non-marked part $w$.

\item If $w\notin P$, the input is accepted.

\item Otherwise $B$ simulates $\bar{A}$ for the marked part of the input and decreases the counter on each step.

\item When the stack is empty and the input is not read fully,
$B$ checks if the current symbol of the input is equal to $\ov{a}$.
If not, then the input is accepted.
\end{enumerate}

\vspace*{-7mm}
\end{proof}

Secondly, we define the \emph{shuffle} (language) of two words $u,v \in \Al^*$ as
follows:
\[
\begin{split}
u\shuffle v=& \{u_1v_1 \cdots u_nv_n \mid u_i, v_i \in \Al^* \text{ for all  } i=1,\dots, n \\
& \text{ and }  u=u_1u_2 \cdots u_n, \ v=v_1v_2 \cdots v_n \}.
\end{split}
\]
In the above factorizations of $u$ and $v$ we allow that some of the factors $u_i$ and $v_i$ are empty.

\medskip
Let $\Al$ and $\Delta$ be two alphabets.
A mapping  $g\colon \Al^* \to \Delta^*$ is a \emph{morphism} if, for all $u,v\in \Al^*$, $g(uv)=g(u)g(v)$.

\medskip
For the undecidability proofs, we use reductions from the
\emph{Post's Correspondence Problem} (PCP, for short).
The PCP was introduced and proved to be undecidable by E. Post in 1946; see
\cite{Po-46}. We shall use the modern form of the problem and define the PCP using monoid morphisms: assume that $g$ and $h$ are two morphisms from
$\Al^*$ into $\Delta^*$, where $\Al=\{a_1,\dots ,a_n\}$ is an alphabet of
$n$ letters.
The pair $(g,h)$ is called an \emph{instance} of the PCP, a word
$w$ satisfying
\begin{equation}\label{vk1}
g(w)=h(w).
\end{equation}
is called a \emph{solution} of the instance
$(g,h)$. The \emph{size} of an instance is the size of the domain
alphabet, i.e., the size is equal to $|\Sigma|$.

\begin{theorem}
It is undecidable for instances $(g,h)$ whether or not it has a solution.
\end{theorem}

It is known that for the size $n\le 2$, the PCP is decidable; see \cite{EKR82} and \cite{HHH02}.
On the other hand, for sizes $n\ge 5$,  the PCP is known to be
undecidable; see \cite{Neary}. The decidability statuses for $n=3$ and $4$ are open.
Note that in basic undecidability proofs, the morphisms
$g$ and $h$ are \emph{non-erasing}, that is, $g(a)\ne \eps \ne h(a)$ for all $a\in \Al$.
This is also the case in \cite{Neary}.

\section{Regular languages}

In this section we study the problems defined in the first section for regular languages.

\subsection{Powers and copies}

Our first theorem is well-known and can be regarded as folklore.

\begin{theorem}\label{Reg:copy}
Let $n \ge 2$ be a fixed integer and $P \subseteq \Al^*$ a regular language.
It is decidable for regular languages $R\subseteq \Al^*$ if there exists a power $w^n \in R$ for some word $w \in P$.
Indeed, the existence of a power $w^n \in R$ is a PSPACE complete problem.
\end{theorem}

\begin{proof}
Let $A$ be a finite automaton accepting $R$, i.e.,
$R=L(A)$.
Let the states of $A$ be $q_0, q_1, \ldots, q_m$, where $q_0$ is the initial state.
Define, for all $i$ and $j$, the regular language $R_{ij}$ by
\[
R_{i, j} = \{w \mid  q_i \xrightarrow{ \ w \ }  q_j \} \cap P,  
\]
where $q_i \xrightarrow{ \ w \ }  q_j$ denotes that there is a computations from the state $q_i$ to the state $q_j$ reading the word $w$ in $A$.
Then there is an $n$th power $w^n$ with $w \in P$ accepted by $A$ if and only if there is
an accepting sequence,
\[
q_0 \xrightarrow{ \ w \ } q_{i_1}  \xrightarrow{ \ w \ }\ldots  \xrightarrow{ \ w \ } q_{i_n},
\]
where $q_{i_n}$ is a final state, i.e., if
$R_{0,i_1} \cap R_{i_1,i_2} \cap  \cdots  \cap R_{i_{n-1},i_n}  \ne \emptyset$.
For each sequence $0, i_1, \ldots, i_n$, the intersection is a regular language.
Moreover, there are only finitely many such sequences of length $n+1$.
Since the emptiness problem
is decidable for regular languages, the claim follows.

\medskip
For PSPACE completeness we need to do the reduction above to the other direction.  Indeed, let $\mathcal{A}_1,\dots, \mathcal{A}_n$ be finite automata accepting languages $L_1,\dots, L_n$. Now construct a new automaton $\mathcal{A}$ by adding transitions reading a new symbol $\#$ from the final states of $\mathcal{A}_i$ to initial state of $\mathcal{A}_{i+1}$ for $i=1,\dots, n-1$, and add new final state $f$ to $\mathcal{A}$, such that from all final states of $\mathcal{A}_n$, there is a transition to $f$ reading $\#$. It is immediate that $L_1\cap\cdots \cap L_n\ne \emptyset$ if and only if there exists $w(=u\#)$ such that $w^n\in L(\mathcal{A})$. The PSPACE completeness now follows from the PSPACE completeness of emptyness of the intersection problem, see~\cite{Rabin}.
\end{proof}

By setting $P=\Al^*$ and $n=2$ in Theorem~\ref{Reg:copy}, we have, see also Anderson et al. \cite{Anderson2009}.

\begin{corollary}\label{Reg:copyspes}
It is decidable for a given regular language $L\subseteq \Al^*$ whether or not there exists $w\in \Al^*$ such that $ww\in L$.
\end{corollary}

The proof of Theorem~\ref{Reg:copy} also gives a well-known result
on the \emph{roots} of words:
the $n$th root of a regular language $R$,
\[
\sqrt[n]{R} = \{ w \mid w^n \in R \},
\]
as well as, the collection of all the roots,
\[
\sqrt[*]{R} =\bigcup_{n \ge 2}\sqrt[n]{R} = \{ w \mid w^n \in R \text{ for  some } n \ge 2 \},
\]
are regular. Let us mention that the regularity  does not hold in the limit case as we see from the following lemma.

\begin{lemma}
For a regular language $R$, the language
\[
\mathit{Pr}(\sqrt[*]{R}) = \{ w \mid w \text{ primitive} \text{ and } w^n \in R \text{ for  some } n \ge 2 \}
\]
is not necessarily regular.
\end{lemma}

\begin{proof}
Let $Q=\mathit{Pr}(\sqrt[*]{\Al^*})$, that is, $Q$ is the
language of all primitive words over $\Al$.
Assume that $Q$ is regular over $\Al=\{a,b\}$,
and consider the (regular) complement $Q^c$ of $Q$ consisting of all non-primitive words
over $\Al$. Then $a^nba^nb \in Q^c$ for all $n$, and thus for sufficiently large $n$, (indeed, larger that $p$ in the Pumping Lemma~\ref{PLreg})
 $a^{n+k}ba^n b\in Q^c$ for some $k \ge 1$ by the Pumping Lemma,
but the word $a^{n+k}ba^n b$ is clearly primitive; a contradiction.

\medskip
Indeed, $Q$ is not even deterministic CF-language; see Lischke~\cite{Lischke2011}.
\end{proof}

For the sake of completeness, we state the following theorem on the inclusion problem $C_P \subseteq L$ for the copy languages  $C_P$.

\begin{theorem}\label{Reg:inclusion}
For regular languages $R$ and $P$, it is decidable if $C_P \subseteq R$ holds.
In particular,
it is decidable if a regular language $R$ is closed under taking squares, i.e.,
if $C_R \subseteq R$.
\end{theorem}

\begin{proof}
The claim follows from Theorem~\ref{Reg:copy}. Indeed,
$C_P \nsubseteq R$ if and only if the complement $R^c$ of $R$ contains a square $ww$ with $w \in P$.
Since the complement of regular language is also regular, the claim follows.
\end{proof}

The technique in the proof of Theorem~\ref{Reg:copy} can also be used for the marked copy problem.

\begin{theorem}\label{Reg:mc}
It is decidable for regular languages $R$,
if $w\ov{w} \in R$ for some~$w$.
\end{theorem}

\begin{proof}
Let $C$ be a finite automaton accepting $R$, i.e.,
$R=L(C)$. As the regular languages are closed under intersection, let $A$ be a finite automaton accepting the language $R\cap \Al^*\ov{\Al}^*$.
Furthermore, let $B$ be a copy of $A$ where all
letters in the transitions are changed from marked to unmarked and vice versa.
Therefore, $L(B)=\{\ov{u}v\mid u\ov{v}\in R\}$. We may assume that the state set of $A$ as well as that of $B$ is $\{q_0, \dots, q_m\}$ and $q_0$ is the initial state.
\eject
Define, for all $i$ and $j$, the regular language $R_{ij}$ by
\[
\begin{split}
R^A_{i, j} &= \{w \mid  q_i \xrightarrow{ \ w \ }  q_j  \text{ in } A\}\\
R^B_{i, j} &= \{w \mid  q_i \xrightarrow{ \ w \ }  q_j  \text{ in } B\}.
\end{split}
\]
Now, there is a word $w\ov{w}\in R$ if and only if for some state $q_j$ and a final state $q_n$,
\[
q_0 \xrightarrow{ \ w \ } q_{j}  \xrightarrow{ \ \ov{w} \ } q_{n},
\]
in $A$ (and $C$), i.e., if
$R^A_{0,j} \cap R^B_{j,n} \cap \Al^* \ne \emptyset$ for some $j$ and $n$.
There are only finitely many intersections of regular languages to be checked, so the claim follows again from the decidability of the emptiness problem
for regular languages.
\end{proof}

For the reverse copy problem requesting if $ww^r \in R$ for some $w$, we first state a more general case.

\begin{theorem}\label{Reg:mirror-n}
Let $k\ge 1$ be fixed.
It is decidable for regular languages $R$ if $w_1 w_1^r \cdots w_k w_k^r \in R$ for some $w_1, \ldots, w_k$.
\end{theorem}

\begin{proof}
Let, for fixed $k$,
\[
E_k=\{w_1 w_1^r \cdots w_k w_k^r \mid w_i \in \Al^*,
i=1,2, \ldots, k\}.
\]	
As a concatenation of $k$ copies of the linear CF-language $E$ from Example~\ref{LanE}, $E_k$ is a (nondeterministic) CF-language. Therefore, also the language
\[
L_k=E_k \cap R
\]	
is a CF-language as context-free languages are closed under intersection with regular languages. Since the emptiness problem is decidable
for context-free languages, the claim follows.
\end{proof}

Setting $k=1$ in the previous theorem yields a result for the reverse copy problem.

\begin{corollary}\label{Reg:mirror}
It is decidable for a regular language $R$,
if $ww^r \in R$ for some word~$w$.	
\end{corollary}

The problem of Theorem~\ref{Reg:mirror-n} turns out to be undecidable
for context-free languages with $k=1$; see Corollary~\ref{CF:mirror}.

\medskip
The proof of Theorem~\ref{Reg:mirror-n} can also be used for the following theorem.

\begin{theorem}\label{Reg:mirror-inf}
It is decidable if a regular language $R$ contains
a word of the form $w_1 w_1^r \cdots w_k w_k^r$ \  for some $w_1, \ldots, w_k$
and $k \ge 1$.
\end{theorem}

\begin{proof}
Let\vspace*{-1.6mm}
\[
L=\bigcup_{k=1}^\infty E_k\,.
\]	
Clearly, $L$ is a CF-language, and as in the proof of Theorem~\ref{Reg:mirror-n}, the language
$L \cap R$ is a CF-language and the claim follows from the decidability of the emptiness problem of CF-languages.
\end{proof}

\subsection{Shuffles}

We begin by defining the language
\[
L_\Al = \bigcup_{w \in \Al^*} w \shuffle \ov{w}\,.
\]

\begin{lemma}\label{shuffle:marked}
The language $L_\Al$ is not a CF-language if $|\Al| \ge 2$.
\end{lemma}

\begin{proof}
Indeed, let $L=L_\Al \cap \Al^*\ov{\Al}^* = \{w\ov{w} \mid w \in \Al^*\}$.
If $L_\Al$ were a CF-language, then $L$ would be a CF-language as $\Al^*\ov{\Al}^*$ is regular.
However, the copy language $C_{\Al^*}$ is a morphic image of $L$, and as the CF-languages are closed under morphic images, that would make $C_{\Al^*}$ a CF-language; a contradiction.\vspace*{-1mm}
\end{proof}

Engelfriet and Rozenberg~\cite[Theorem 15]{EngRoz} showed in 1980 that each recursively enumerable
language $K$ can be represented in the form
$K=h(L_\Sigma \cap M)$, where $h$ is a letter-to-letter morphism, $\Al$ a binary alphabet and
$M$ a regular language. It follows, as stated in~\cite{EngRoz}, that
it is undecidable for regular languages $R \subseteq \Al^*$ if  $R$ contains
a word from $w \shuffle \ov{w}$ for some $w \in \Al^*$;
see also~\cite{HHS} for a direct proof of this along different lines.

We give a simple proof of the result based on reduction to the Post Correspondence Problem.
The proof below does not apply to small alphabets as the PCP is known to be undecidable only for alphabet sizes of at least five.


\begin{theorem}\label{Reg:marked}
It is undecidable for regular languages $R$ if $R$ contains an element of
$w \shuffle \ov{w}$ for some~$w$.
\end{theorem}

\begin{proof}
Let $(g, h)$ be an instance of the PCP for $g, h\colon \Al^* \to \Delta^*$, where $g$ and $h$ are both non-erasing.
Define a (generalized) finite automaton $A$ with the set of states
\[
Q = \{ v \mid  v \ \text{ is a prefix of $h(a)$, $ a \in \Al$} \}.
\]
We set the state $q_f=\eps$ to be the initial and the unique final state.
The transitions of $A$ are of the form: for $a \in \Al$ and $z \in \Al^+$
\begin{align*}
&u \xrightarrow{ \ a\ov{z} \ } v \ \text{ if there exists } x\in \Al^* \text{ and }v\in Q \text{ such that }g(a)=xv \text{ and $h(z)=ux$}, \\
&u \xrightarrow{ \ a \ } vg(a), \ \text{ if } vg(a)\in Q.
\end{align*}
Note that the words $z$ in the above can be found algorithmically, as the images $g(a)$ are of finite length and $h$ is non-erasing.

\medskip
The states $w\in Q$ correspond to \emph{overflows} of the instance $(g,h)$
when $g$ is leading:
an overflow $w \in Q$ occurs in the situation, where
$g(u)=h(v)w$ for some words $u$ and $v$. Indeed, if the automaton $A$ is in state $w$ after reading the word $v$, then necessarily $v=a_1\ov{z_1}a_2\ov{z_2}\cdots a_n\ov{z_n}$, where $a_i\in \Al$ and $z_i\in \Al^*$ (note that $z_i$ is empty for the transitions of the latter form) for all $i=1,\dots, n$, and
$$
g(a_1)g(a_2)\cdots g(a_n) = h(z_1)h(z_2)\cdots h(z_n)w.
$$
As $q_f=\eps$, we have that $A$ accepts the language
\begin{align*}
L(A) = \{a_1\ov{z_1} a_2 \ov{z_2} \cdots a_n\ov{z_n} \mid
& \ n \ge 0, \  a_i \in \Al, \ z_i \in \Al^*, \\
& \ g(a_1a_2 \cdots a_n)
=h(z_1z_2 \cdots z_n) \}.
\end{align*}
Therefore,  $(w \shuffle \ov{w}) \cap L(A) \ne \emptyset$ for some $w$ if and only if
there exists a word
$a_1\ov{z_1} a_2 \ov{z_2} \cdots a_n\ov{z_n}$ in $L(A)$ such that
$w=a_1a_2 \cdots a_n=z_1z_2 \cdots z_n$.
This is equivalent to saying that the instance $(g,h)$ of the PCP
has a solution.
\end{proof}

The self-shuffle and the shuffle with reverse are left as open problems.

 \begin{problem}
 Is it decidable for regular languages $R$ if $R$ contains an element of
$w\shuffle w$ for some $w$?
 \end{problem}

 \begin{problem}
 Is it decidable for regular languages $R$ if $R$ contains an element of
$w\shuffle w^r$ for some $w$?
 \end{problem}

\EXTRA{%
Let $\varphi\colon (\Al \cup \ov{\Al})^*  \to \Al^*$ be a morphism
defined by $\varphi(a) = a = \varphi(\bar{a})$. 	
Now, for $R \subseteq \Al^*$, we have $(w \shuffle w) \cap R \ne \emptyset$ if and only if $(w \shuffle \ov{w}) \cap \varphi^{-1}(R)
\ne \emptyset$.
However, not every (regular) language
$L \subseteq (\Al \cup \ov{\Al})^*$
is of the form $L=\varphi^{-1}(L_0)$ for some $L_0 \subseteq \Al^*$;
so that Theorem~\ref{Reg:marked} does not give any solution to the problem.
}%

The problems with shuffles of words tend to be algorithmically difficult.
Indeed, the following was shown by Biegler and McQuillan~\cite{Biegler:2014}.

\begin{theorem}
Consider an instance consisting of a DFA $A$ and two words $u,v \in \Al^*$ with $|\Al| \ge 2$.
It is NP-complete to determine if there exists a word $w \in L(A)$ such that
$w \notin u \shuffle v$.
\end{theorem}

\begin{problem}
Does the above problem stay NP-complete if an instance consists of $A$ and a single word $u$,
and the problem is to determine if there exists a word $w$ with $w \notin u \shuffle u$?
\end{problem}

\section{Linear CF-languages}

In this section we study the problems defined in the introduction for linear CF-languages and show that they are all undecidable.

\subsection{Powers and copies}

Let $(g, h)$ be an instance of the PCP, where $g,h\colon \Gamma^* \to \Delta^*$ with $\Gamma \cap \Delta = \emptyset$. Define the language
\[
L_2(g,h) = \{z u^r x w^r  :  u, w \in \Gamma^+, \ z,x \in \Delta^*, \ z=g(w), \ x=h(u)\}.
\]
It is an easy exercise of language theory to show that $L_2(g,h)$ is accepted by a deterministic linear pushdown automaton. Indeed, the reversal happens between $u^r$ and $x$ as the automaton can recognize the alternation in the
disjoint alphabets.
Checking if $z=g(w)$ and $x=h(u)$ can be easily performed deterministically
while decreasing the stack after the reversal. Therefore, $L_2(g,h)$ is a linear CF-language.

\begin{theorem}\label{Lin:copy}
It is undecidable for deterministic linear CF-languages $L$ if
$ww \in L$ for some~$w$.
\end{theorem}

\begin{proof}
Suppose $g,h \colon \Gamma^* \to \Delta^*$ are morphisms where
$\Gamma$ and $\Delta$ are disjoint,
and let $\Sigma= \Gamma \cup \Delta$.
We rewrite the above language in a more convenient form:
\begin{equation}\label{eq:lin}
L_2(g, h) = \{g(w) u^r h(u) w^r  :  u, w \in \Gamma^+\}.
\end{equation}
Obviously, there is a square in $L_2(g,h)$  if and only if $g(w)u^r=h(u)w^r$ for some non-empty words $u$ and $w$, that is, if and only if the instance $(g,h)$ of the PCP has a nonempty
solution: $g(w)=h(u)$ and $w=u$.

\medskip
Since $L_2(g, h)$ is a deterministic linear CF-language, the claim follows.
\end{proof}

Since the deterministic linear CF-languages are closed under taking complements, we have, corresponding to Theorem~\ref{Reg:inclusion}, the following corollary.

\begin{corollary}
It is undecidable for deterministic linear CF-languages
$L$ if all squares are in $L$, i.e.,
if $C_{\Al^*} \subseteq L$.
\end{corollary}

\begin{proof}
We consider the negation of the proposition in Theorem~\ref{Lin:copy}:
\[
\neg \exists w: ww \in L
\iff
\forall w:  ww \notin L
\iff
\forall w: ww \in L^c.
\]

\vspace*{-8mm}
\end{proof}

For the marked copy, we transform the language $L_2(g,h)$ in \eqref{eq:lin} into the language
\[
\bar{L}_2(g, h) = \{g(w) u^r \ov{h(u)} \ov{w}^r  :  u, w \in \Gamma^*\}
\]
which is clearly also a deterministic linear CF-language.

\begin{theorem}\label{Lin:marked}
It is undecidable for deterministic linear CF-languages $L$, if
$w\ov{w} \in L$ for some word~$w$.	
\end{theorem}

\begin{proof}
As in the proof of Theorem~\ref{Lin:copy}, because the alphabets $\Delta$ and $\Gamma$ are disjoint (therefore, so are $\ov{\Delta}$ and $\ov{\Gamma}$)
we obtain that there is a word of the form $w\ov{w}$ in $\bar{L}_2(g, h)$ if and only if $g(w)u^r=h(u)w^r$ which again is equivalent to saying
that the instance $(g,h)$ of the PCP has a solution.
\end{proof}

\begin{example}
The language
\[
F = \{w \in \Sigma^* \mid w= uxxv \ \text{ for some nonempty word } x \}
\]
is not regular. Indeed, by the Pumping lemma, its complement,
the language of all square-free words,  is not even context-free. For this, assume contrary that $F^c$ is a CF-language. As it is infinite, pumping in
Lemma~\ref{PLCF},
implies that for every sufficiently long word $w \in F^c$ contains a factor of
the form $xx$ for some nonempty word $x$. Therefore, there is a word of the form $uxxv$ in $F^c$; a contradiction.
\end{example}

We extend Theorem~\ref{Lin:copy} for arbitrary powers $n\ge 2$ as follows.

\begin{theorem}\label{Lin:pow}
Let $n \ge 2$ be a fixed integer.
It is undecidable for deterministic linear CF-languages $L \subseteq \Al^*$
if there exists a power $w^n \in L$ for some
nonempty $w \in \Al^*$.
\end{theorem}

\begin{proof}
Our construction relies on the language $L_2(g,h)$ defined in \eqref{eq:lin}.
Let
\[
L_n(g,h)=  L_2(g,h) \cdot (\#\Delta^+\Gamma^+)^{n-2}.
\]
It is a deterministic linear CF-language since the end portion $(\Delta^+\Gamma^+)^{n-2}$ is regular. Recall that $\Delta\cap\Gamma=\emptyset$.

\medskip
Now,\vspace*{-1mm}
\[
g(w) u^r h(u) w^r\cdot  w_1u_1 \cdots w_{n-2}u_{n-2}\in L_n(g,h)
\]
is an $n$th power if and only if $g(w)=h(u)=w_1 = \ldots =w_{n-2}$ and
$u^r=w^r =u_1 = \ldots =u_{n-2}$, and thus
if and only if the instance $(g,h)$ has a solution. This proves the claim.
\end{proof}

Finally, if we replace $L_n(g,h)$ by the deterministic linear CF-language
\[
L_\omega(g,h)=  L_2(g,h) \cdot (\Delta^+\Gamma^+)^*,\vspace*{-1mm}
\]
we have,

\begin{theorem}\label{Lin:powers}
It is undecidable for deterministic linear CF-languages $L \subseteq \Al^*$
if there exists a power $w^n \in L$ for some
$n\ge 2$ and nonempty $w \in \Al^*$.
\end{theorem}

Regarding the decidability result in Theorem~\ref{Reg:inclusion} for regular languages, we state an open problem for linear CF-languages.

\begin{problem}\label{Prob:closure}
Is it decidable for linear CF-languages $L$ if
$C_L \subseteq L$, i.e., if $L$ is closed under taking squares?
\end{problem}

In many special cases the answer to the above problem is positive. Indeed, according to Greibach~\cite{Greibach:linear}
if a language $L_1cL_2$, with $L_i \subseteq (\Al \setminus \{c\})^*$ for $i=1,2$, is a linear CF-language then $L_i$ is regular for $i=1$ or $i=2$.

The inclusion problem of squares in a language becomes undecidable ''just above'' the regular languages. Recall that counter languages are accepted with
(nondeterministic) pushdown automata with a single pushdown letter for the stack.

\begin{theorem}\label{Counter:inclusion}
It is undecidable for counter languages $L \subseteq \Al^*$
if $C_{\Al^*} \subseteq L$ holds.
\end{theorem}

\begin{proof}
We prove the claim by reduction from the PCP. Let $(g,h)$ be an instance of the PCP
for $g,h\colon \Al^* \to \Delta^*$. Let $\Gamma = \Al \cup \{\#\}$, where $\#$ is a new letter.

\medskip
We now describe a (nondeterministic) counter language $L \subseteq \Gamma^*$
such that $w \in L$ if
\begin{enumerate}
\item[(1)] $w \ne u\#v\#$,
for all $u,v \in \Al^*$, or
\item[(2)] $w=u\#v\#$ but $u \ne v$ or
$g(u) \ne h(v)$.
\end{enumerate}
Clearly, $L$ contain all squares of $\Gamma^*$ except those words $w\#w\#$ with $g(w)=h(w)$.
Thus $C_{\Al^*} \subseteq L$ if and only if  the instance $(g,h)$ has no solutions, and the claim follows from the undecidability of the PCP.

Starting from its initial state the automaton $M$ branches to one of the three separate
lines of  actions,
and it accepts the input $w$ if one of these lines leads to acceptance.
Note that counter automata are nondetermistic, and their languages are closed under union.

\medskip
 The automaton $M$ accepts if
 \begin{itemize}
 \item[1.]
$w \notin \Al^*\# \Al^* \#$. Since  $\Al^*\# \Al^* \#$ is regular this can be
decided with the states, without counter actions.

From this on, we suppose that $w=u\#v\#$, where $u,v \in \Al^*$.

\item[2.]  $u \ne v$.  While reading the prefix $u$ and increasing the counter,
the automaton guesses a position
$n \le |u|$, say with the letter $a$. The counter stack has then $c^n$.
The automaton reads on until it reaches the first $\#$,
after which it pops the counter to gain the $n$th letter, say $b$, of $v$.
It accepts if $a \ne b$, that is the $n$th letter of $u$ is not equal to $n$th letter of $v$.

\item[3.]  $g(u) \ne h(v)$. As in case 2 the automaton can guess a position of different letter while in the images $g(u)$ and $h(v)$.
E.g., while reading~$u=a_1\cdots a_m$, the automaton
guesses a position $n=|g(a_1a_2 \cdots a_k)|+j$ by pushing
$c^{|g(a_i)|}$ $ i=1,2, \ldots, k$, to the counter when reading $a_1,\dots, a_k$ and then pushing $c^j$, for $j < |g(a_{k+1})|$, and remembers the symbol $a$ which is the $i$th symbol of     $g(a_{k+1})$ and, therefore, the $n$th symbol of $g(u)$.

Then, when reading $v=b_1\cdots b_t$, $M$ decreases the counter with $c^{|h(b_i)|}$ until $|h(b_1\cdot b_\ell)|+s=n$ for some $\ell$ such that $s< |h(b_{\ell +1})|$ and checks that the symbol in position $s$ of $h(b_{\ell +1})$, that is, the  $n$th symbol of $h(v)$ are different.
\end{itemize}
It is straightforward to see that $M$ accepts the language $L$.
\end{proof}

\subsection{Shuffles}

We begin with the shuffle operation on CF-languages.
By considering the generating grammars, it is evident that
the family of context-free languages is closed under
concatenation. However, it is not closed under shuffle of languages.
To see this, consider the languages
\[
L_1 = \{a^nb^n \mid n \ge 1 \}, \
L_2 = \{c^ma^m \mid m \ge 1  \} \ \text{ and } \
L=L_1 \shuffle L_2 \cap a^*c^*b^*a^*.
\]
If $L_1\shuffle L_2$ is a CF-language, then so is $L$ as CF-languages are closed under intersection with regular languages and $a^*c^*b^*a^*$ is regular. But $L = \{a^nc^mb^n a^m \mid n,m \ge 1\}$ is a well-known non-CF-language.
Note that the languages $L_1$ and $L_2$ are even
deterministic linear CF-languages.

\medskip
The following example showing that CF-languages over binary alphabets are not closed under shuffle is due to Chris K\"ocher \cite{CKoc}.

\begin{example}
Let $L_1=\{ a^n b a^n\mid n\ge 1\}$ and $L_2=\{ b^n a b^n\mid n\ge 1\}$. Now
$$
(L_1\shuffle L_2)\cap a^*b^*a^*b^* = \{a^mb^{n+1}a^{m+1}b^n\mid m,n\ge 1\},
$$
which is not a context-free language. As the language $a^*b^*a^*b^*$ is regular, and intersection of a context-free language and a regular language is always context-free, we get that  $L_1\shuffle L_2$ is not a context-free language.
\end{example}

Next, we study the shuffle problem for linear CF-languages.

\begin{theorem}\label{Lin:shuffle}
It is undecidable for deterministic linear CF-languages $L$ if $L$
contains an element of $w \shuffle w$ for some $w$.
\end{theorem}

\begin{proof}
Recall the language  $L_2(g,h)$ from~\eqref{eq:lin}, and consider its modification
\begin{equation}\label{eq:lin2}
L_\#(g,h) = \{\$ g(w) u^r\# \$ h(u) w^r \# :  u, w \in \Gamma^*\},
\end{equation}
where the markers $\$$ and $\#$ appears only in the given positions. Clearly, $L_\#(g,h)$ is deterministic linear CF-language.

\medskip
The only word of the form $\$ u\#\$ v\#$ in the shuffle  $\$ w\# \shuffle \$ w\#$ is $\$ w\# \$ w\#$, which in its turn belongs to $L_\#(g,h)$ if and only if
the instance $(g,h)$ has a solution, the claim follows.
\end{proof}

\begin{corollary}
It is undecidable for deterministic linear CF-languages
$L \subseteq \Al^*$ if
all shuffles $w \shuffle w$ for $w \in \Al^*$ are in $L$.	
\end{corollary}

\begin{proof}
We consider the negation of the proposition in Theorem~\ref{Lin:shuffle}:
\[
\neg \exists w: w \shuffle w \cap L \ne \emptyset
\iff
\forall w:  w \shuffle w \cap L = \emptyset
\iff
\forall w: w \shuffle w \subseteq L^c.
\]
Since the family of deterministic linear CF-languages is closed under complement, the claim follows.
\end{proof}

By Rizzi and Vialette~\cite{RizziVialette}
and Buss and Soltys~\cite{BussSoltys}, it is an NP-complete problem
if a word $v$ is a self-shuffle, i.e., if there exists a word $w$ such that
$v=w \shuffle w$.

\medskip
Consider the shuffle $w \shuffle w^r$, where $w^r$ is the reverse of $w$.
For a language $P$, let
\[
M_P = \bigcup_{w \in P} w \shuffle w^r.
\]

See Henshall et al.~\cite{Henshall} for the following result.

\begin{theorem}\label{rev:CF}
The language $M_{\Al^*}$ is not context-free.
\end{theorem}

We give a bit simpler result related to the previous theorem as an example.

\begin{example}
We show that the language  $M_P$ need not be regular for regular $P$.
Indeed, let $P=a^+b^+$ over $\Al=\{a,b\}$.
Then $L=M_P\cap a^*b^*a^*= \{a^n b^{2m} a^n\mid n,m\ge 1\}$ is non-regular (by the Pumping Lemma). Therefore, $M_P$ is not regular.

\medskip
Similarly, let $L=\{a^nb^n\mid n\ge 1\}$. Now,
\[
M_L\cap a^*b^*a^*=
\bigcup_{n \ge 1}\left(a^n b^n \shuffle b^n a^n\right)
\cap a^*b^*a^* = \{a^n b^{2n} a^ n \mid n \ge 1\}
\]
is not context-free, and, therefore,  $M_L$ is not context-free,
since $a^*b^*a^*$ is regular.
\end{example}

Our next theorem concerns linear CF-languages that need not
be deterministic.

\begin{theorem}\label{Lin:revshuf}
It is undecidable for linear CF-languages
if $L$ contains an element of $w \shuffle w^r$ for some~$w$.
\end{theorem}

\begin{proof}
The proof is by reduction from the PCP.
Let $(g, h)$ be an instance of the PCP with
$g, h\colon \Sigma^* \to \Delta^*$ and let  $\#$ be a new symbol $\# \notin \Delta$.
Consider the linear CF-language $L_1= L(G)$
generated by the grammar $G$ with two non-terminals $S$ and $T$
together with the production rules
\begin{align*}
S &\to \#g(a)T h(a)^r\# \ \text{ for all $a\in \Sigma$},\\
T &\to  g(a)T h(a)^r \ \text{ for all $a\in \Sigma$},\\
T &\to \#\#.
\end{align*}
Hence
\begin{equation}\label{eq:wwr}
L_1 = \{\#g(v)\#\# h(v)^r\# \mid v \in \Sigma^+ \}.
\end{equation}
Now, $L_1$ contains an element of $w \shuffle w^r$ for a word $w$
if and only if $w=\#g(v)\#=\#h(v)\#$ for some nonempty word $v$
for which then $g(v)=h(v)$.
The claim follows from the undecidability of the PCP.
\end{proof}

The linear CF-language $L_1$ in \eqref{eq:wwr}
also gives immediately the following result.

\begin{theorem}\label{Lin:revcopy}
It is undecidable for linear CF-languages $L$, if
$ww^r \in L$ for some word $w$.	
\end{theorem}

The proof of Theorem~\ref{Lin:revshuf} also gives the following corollary.

\begin{corollary}\label{CF:mirror}
It is undecidable for linear CF-languages $L$
if $L$ contains a palindrome, i.e, a word $w$ such that  $w=w^r$.
\end{corollary}

The following result is a trivial consequence of Theorem~\ref{Reg:marked},
which stated the result already for regular languages.

\begin{theorem}\label{Lin:marshuff}
It is undecidable for linear context-free languages
if $L$ contains an element of $w \shuffle \ov{w}$ for some $w$.
\end{theorem}

\EXTRA{
A context-free language $L$ is  said to be \emph{$k$-turn} if it is accepted by a pushdown automaton that makes at most $k$ turns, from increase to decrease mode, in its stack.

\begin{theorem}\label{undec:rev}
It is undecidable for deterministic $2$-turn context-free languages
if $L$ contains an element of $w \shuffle w^r$ for some $w$.
\end{theorem}

\begin{proof}
Let again $(g, h)$ be an instance of the PCP with
$g, h\colon \Gamma^* \to \Delta^*$ where the alphabets are disjoint,
and $\Delta$ does not contain the letter $\#$.
Let $\Al=\Gamma \cup \Delta \cup \{\#\}$. Consider the set
\begin{equation}\label{eq:rev}
M_1(g, h) = \{ug(u)^r \#\# h(v)v^r \mid u, v \in \Gamma^* \}.
\end{equation}
The language $M_1(g, h)$ is deterministic $2$-turn context-free.
Moreover, $M_1(g, h)$ has an element $ug(u)^r \#\# h(v)v^r \in w  \shuffle w^r$ if and only if
$u=v$ and $g(u)=h(v)$, i.e., iff the instance $(g, h)$ has a solution.
\end{proof}

The above theorem and its proof can be reformulated as follows.

\begin{corollary}
It is undecidable for two deterministic linear CF-languages
$L_1$ and $L_2$ if $L_1L_2$ contains an element of $w\shuffle w^r$
for some $w$.
\end{corollary}
}

In contrast to the decidability result in Theorem~\ref{Reg:mirror-n}
we can extend the proof of Theorem~\ref{Lin:revshuf} for the next claim.

\begin{theorem}\label{CF:powers}
Let $k\ge 1$ be fixed.
It is undecidable for CF-languages $L$ if $L$ contains a word of the form $w_1 w_1^r \cdots w_k w_k^r$ for some $w_1, \ldots, w_k$.
\end{theorem}

\begin{proof}
We modify the linear CF-language from \eqref{eq:wwr}.
Consider the 'next' linear CF-language $L_k$ for the instance $(g, h)$:
$
L_k = L_1 \cdot (cc)^{k-1},
$
where $c$ is a new letter.
\end{proof}

For regular languages $R$ it is clearly decidable if $\{w, w^r \}^* \subseteq R$ for given $w$.

\begin{problem}
Is it decidable for regular languages $R$ if $\{w, w^r \}^* \subseteq R$ for some $w$?
How about CF-languages?	
\end{problem}

On the other hand, it is decidable for context-free languages $L$ if $L \subseteq \{w, w^r \}^*$ for some $w$.
Indeed, by the Pumping property, the length of possible words $w$ has an
effective upper bound.

\section{Conclusions}

Our aim was to present a survey on the decidability statuses of special membership problems for copies, marked copies, reversed copies, self-shuffles and shuffles with marked or reversed copies. Two cases, the self-shuffle and shuffle with the reversed copy, remain unsolved. We also studied special inclusion problems regarding powers and especially squares of words. Several related open problem were stated.

\end{document}